\documentclass[a4paper,11pt]{article}
\usepackage{amssymb}

\usepackage{amsmath, amsthm}
\usepackage{enumerate}
\usepackage{url}

\usepackage[ruled,linesnumbered,vlined]{algorithm2e}

\newtheorem{theorem}{Theorem}
\newtheorem{definition}[theorem]{Definition}
\newtheorem{lemma}[theorem]{Lemma}

\newtheorem{corollary}[theorem]{Corollary}

\newtheorem{observation}[theorem]{Observation}
\newtheorem{proposition}[theorem]{Proposition}

\newcommand{\pref}{\succ}

\newcommand{\agent}{a}

\newcommand{\matching}{\mu}

\newcommand{\inputsetting}{I}

\newcommand{\eqclass}{C}
\newcommand{\minicr}{{\sc Min-ICR}}
\newcommand{\icr}{{\sc ICR-Dec}}
\newcommand{\smpi}{{\sc SMPI}}

\newcommand{\smi}{{\sc SMI}}
\newcommand{\smt}{{\sc SMT}}
\newcommand{\smti}{{\sc SMTI}}
\newcommand{\minicrx}{{\sc Min-ICR-Exact}}
\newcommand{\icrx}{{\sc ICR-Exact-Dec}}
\newcommand{\vc}{{\sc VC}}
\newcommand{\minvc}{{\sc Min-VC}}
\newcommand{\vct}{{\sc VC-3}}
\newcommand{\manset}{M}
\newcommand{\womanset}{W}
\newcommand{\numofmen}{n}
\newcommand{\numofwomen}{n}
\newcommand{\man}{m}
\newcommand{\woman}{w}
\newcommand{\partialorder}{p}

\newcommand{\candidate}{c}
\newcommand{\interviewset}{T}

\providecommand{\myceil}[1]{\left \lceil #1 \right \rceil}

\newcommand{\hide}[1]{}
\begin{document}
\title{\bf Preference Elicitation in Matching Markets via Interviews: A Study of Offline Benchmarks\thanks{This research has been supported by EPSRC grants EP/K01000X/1 and EP/K010042/1. The authors gratefully acknowledge the support of COST Action IC1205 on Computational Social Choice. The authors would also like to thank Rob Irving and Piotr Krysta for the useful discussions and their valuable feedback.}
}

\author{Baharak Rastegari$^1$, Paul Goldberg$^2$, David Manlove$^1$\\
\\
\small
\emph{$^1$ School of Computing Science, University of Glasgow, Glasgow, UK}
\\
\small
\emph{Email: {\tt \{baharak.rastegari,david.manlove\}@glasgow.ac.uk}.}
\\
\\
\small
\emph{$^2$ Department of Computer Science, University of Oxford, Oxford, UK}\\
\small
\emph{Email: {\tt paul.goldberg@cs.ox.ac.uk}.}}
\date{}	

\maketitle              

\begin{abstract}
The {\em stable marriage problem} and its extensions have been extensively studied, with much of the work in the literature assuming that agents fully know their own preferences over alternatives. This assumption however is not always practical (especially in large markets) and agents usually need to go through some costly deliberation process in order to learn their preferences.  In this paper we assume that such deliberations are carried out via \emph{interviews}, where an interview involves a man and a woman, each of whom learns information about the other as a consequence.  If everybody interviews everyone else, then clearly agents can fully learn their preferences.  But interviews are costly, and we may wish to minimize their use.  It is often the case, especially in practical settings, that due to correlation between agents' preferences, it is unnecessary for all potential interviews to be carried out in order to obtain a stable matching.  Thus the problem is to find a good strategy for interviews to be carried out in order to minimize their use, whilst leading to a stable matching.  One way to evaluate the performance of an interview strategy is to compare it against a na\"ive algorithm that conducts all interviews. We argue however
that a more meaningful comparison would be against an optimal offline algorithm that has access to agents' preference orderings under complete information. We show that, unless P=NP, no offline algorithm can compute the optimal interview strategy in polynomial time.  If we are additionally aiming for a particular stable matching (perhaps one with certain desirable properties), we provide restricted settings under which efficient optimal offline algorithms exist. 
\end{abstract}


%
\section{Introduction}
\label{sec:intro}
{\em Two-sided matching markets} model many practical settings, such as corporate hiring and university admission \cite{two-sided-matching-book,manlove-book}.  The classical {\em stable marriage problem} is perhaps the most widely studied matching problem in this class, where participants are partitioned into two disjoint sets -- men and women -- and each participant on one side of the market wishes to be matched to a candidate from the other side of the market, and has preferences over potential matches.  A matching is called stable if no pair of participants would prefer to leave their assigned partners to pair with each other. Gale and Shapley's seminal paper~\cite{gale-shapley-62} proposed a polynomial-time algorithm for finding a stable matching. The books by Knuth~\cite{knuth-french}, Gusfield and Irving~\cite{irving-gusfield}, Roth and Sotomayor~\cite{two-sided-matching-book}, and Manlove~\cite{manlove-book} provide excellent introductions and surveys.

A key assumption in much of this literature is that all market participants know their full 
preference orderings. The classical Gale-Shapley (GS) algorithm~\cite{gale-shapley-62} and its variants require participants' preferences as input. This assumption is reasonable in some settings. However, as markets grow large (e.g., in the hospital-resident matching market~\cite{roth-nrmp,irving-manlove-09} or college admission market~\cite{gale-shapley-62,Roth-85}) it quickly becomes impractical for participants to assess their precise preference rankings. Instead, participants usually start out with some partial knowledge about their preferences and need to perform some deliberation in order to learn their precise preference ordering. In this paper we assume that deliberations are carried out via \emph{interviews}, where an interview is a unit operation that involves one agent from each side of the market and is informative to both participants. For example, in the hospital-resident problem (which models the entry-level labor market in which graduating medical students, or residents, are seeking to be assigned to hospital posts), hospitals are likely to be able to identify their ``top-tier'' residents, ``second-tier'' residents and so on, and in order to rank the residents in each tier they need to interview them. An interview between a hospital $h$ and a resident $r$ yields information about the qualities of each party to the other.  Thus we initially assume that each agent's preference list is in general expressed in terms of a partial order, and after an agent has interviewed $\ell$ members of the opposite side of the market, he/she has discovered enough information to rank those elements in strict order.

In order to be able to use the GS algorithm to find a stable matching in this setting, a na\"ive solution is for each participant to conduct all potential interviews and fully learn their preferences. Interviews however are usually costly both in terms of time, mental energy, and money, therefore we wish to minimize their usage. Indeed, the na\"ive approach may impose unnecessary deliberation. For example, in the hospital-resident problem, one expects some degree of correlation across hospitals in the assessment of the most desirable residents, and likewise residents are expected to have correlated views (at least to some extent) on the desirability of hospitals. Therefore, it is expected that more desirable residents get matched to more desirable hospitals and so on. It is then not hard to see that it is wasted effort if a top-tier resident is to interview low-tier hospitals, or a second-tier resident is to interview top-tier hospitals. For a concrete example, consider a setting with four residents and four hospitals where each hospital can admit (at most) one resident. Assume that residents $r_1$ and $r_2$ are top-tier residents and $r_3$ and $r_4$ are second-tier residents. Likewise assume that $h_1$ and $h_2$ are top-tier hospitals and $h_3$ and $h_4$ are second-tier hospitals. The preference lists of all agents are correlated according to these hierarchies, although each agent's individual strict ranking (initially unknown) within these hierarchies may differ.  It is not hard to verify that no matter what the true (initially unknown) preference orderings of the participants are, under a stable matching $r_1$ and $r_2$ each gets matched to either $h_1$ or $h_2$, and $r_3$ and $r_4$ each gets matched to either $h_3$ or $h_4$. Thus an interview between $r_1$ and either $h_3$ or $h_4$ is unnecessary, for example.
%

Unfortunately, we cannot always avoid unnecessary interviews. For example consider a setting with two residents and two hospitals, where initially agents have no information on their preference orderings and hence cannot compare the two alternatives. W.l.o.g.\ assume that $h_1$ interviews both residents and learns that it prefers $r_1$ to $r_2$. If $r_1$ additionally interviews $h_2$ and learns that he prefers $h_1$ to $h_2$, then a stable matching $\matching$ is found after 3 interviews, in which $h_i$ is matched to $r_i$ ($1\leq i\leq 2$). Now imagine that $r_1$ instead learns that he prefers $h_2$ to $h_1$. It is easy to verify that the identity of a stable matching is not yet revealed and hence more interviews are required. The only remaining interview is between $h_2$ and $r_2$ after which one can definitely identify a stable matching. If $h_2$ learns that it prefers $r_2$ to $r_1$ then $\matching$ is a stable matching. Assume that $r_2$ also prefers $h_2$ to $h_1$. In this case the interview between $h_1$ and $r_1$ is unnecessary as the three other interviews would have provided enough information -- that $h_2$ and $r_2$ are each others' top choice -- for $\matching$ to be identified as a stable matching. However, a priori we can not always rule the interview between $h_1$ and $r_1$ as unnecessary; for example when $h_2$'s top choice is $r_2$ but $r_2$'s top choice is $h_1$.

Any interviewing strategy leads to refinements of the partial orders contained in the original problem instance that represented uncertainty over the true preferences.  A key aim could be to carry out sufficient interviews so as to arrive at an instance that admits a super-stable matching $\mu$.  Super-stability will be defined formally in the next section, but informally it ensures that $\mu$ will be stable regardless of how the remaining uncertainty is resolved.  
The original instance need not admit a super-stable matching (see \cite{Irving94} for an example) but we are guaranteed that a super-stable matching is always achievable (e.g., by conducting all possible interviews, we will arrive at a strictly-ordered instance, where super-stability and classical stability become equivalent, and the existence of a stable matching is assured \cite{gale-shapley-62}).

Thus our aim is to find a good strategy that conducts as few interviews as possible so as to obtain a refined instance that admits a super-stable matching.  In general any such strategy will be an online algorithm, since the next interview to be carried out might depend on the results of previous ones.

This leads to the question of how to evaluate the performance of any given interview strategy. One could compare it against the na\"ive algorithm described above that conducts all interviews. We argue however, by analogy with online algorithms and their competitive ratio, that it makes more sense to compare it against an optimal ``offline'' algorithm.  Here, the optimal offline algorithm has access to agents' preference orderings under full information and has to compute the optimal (i.e., minimum) number of interviews required in order to reach an information state under which it can identify a super-stable matching. In this paper we show that unless P=NP, no offline algorithm can compute an optimal interview strategy in polynomial time.

Some stable matchings have desirable properties, and we may be interested in refining the preferences further so as to obtain such matchings.  For example, in the \emph{man-optimal stable matching}, each man has the best partner that he could obtain in any stable matching, whilst the \emph{woman-optimal stable matching} has a similar optimality property for the women.  As described above, after a certain number of interviews we may reach an instance that admits a super-stable matching $\mu$.  But by carrying out more interviews, some men, for example, may end up with better partners than they had in $\mu$.  This would be the case if $\mu$ is not the man-optimal stable matching in the instance with the strict (true underlying) preferences.

If we wish to evaluate the performance of an online algorithm that aims for potential improvements in men's partners even after a super-stable matching has been identified, then a suitable offline benchmark for the competitive ratio would be the minimum number of interviews required to refine the original instance so as to make a \emph{specified} matching super-stable.  In this paper we show that, whilst this problem is NP-hard in general, there are restricted cases that are solvable in polynomial time.
  

\paragraph{Related work} 
Until very recently, the problem of incremental preference elicitation has received little attention. Several works in the past few years however have addressed this problem from different angles~\cite{KU08,schwarz,DB-ijcai13,federico-extremal,RCIL-ec13,gium,DB-aaai14}. Those closest in spirit to ours are \cite{DB-ijcai13,RCIL-ec13,DB-aaai14}. 

In \cite{RCIL-ec13} the authors introduced a stable matching model in which participants start out with incomplete information about their preferences, in the form of partially ordered sets, and are able to refine their knowledge by performing interviews. They investigated the problem of minimizing the number of interviews required to find a matching that is stable w.r.t.\ the true underlying strict preference ordering and additionally is optimal for one side of the market. They presented several results among which are the following two: (i) finding a minimum certificate, that is a set of partial preferences that supports an optimal (for one side of the market) stable matching is NP-hard, and (ii) in a setting where participants on one side of the market have the same partially ordered preferences, an optimal interview policy can be found in polynomial time.

In \cite{DB-ijcai13}, the authors studied a setting where deliberation is in the form of pairwise comparison queries (that is, a query leads to strict order of preference being determined over two acceptable agents for a given agent). They proposed a method for finding approximately stable matchings, using minimax regret as a measure, while keeping the number of required comparisons relatively low. In \cite{DB-aaai14} the authors combined the comparison query model  of \cite{DB-ijcai13} with the interview model of \cite{RCIL-ec13} and introduced a unified model where both types of elicitation can take place. They provided an efficient (polynomial-time) scheme for generating queries and interviews, and examined the effectiveness of their scheme via empirical evaluation including comparison against the polynomial-time algorithm of \cite{RCIL-ec13} for the restricted setting in which participants on one side of the market have the same partially ordered preferences.

Our work is also related to the body of literature studying variants of stability  defined in settings where agents' preferences may include ties. As discussed above, super-stable matchings are relevant in the context of incomplete preference information, because they are stable no matter which refinements represent the true (strict) preferences.  
Polynomial time algorithms have been proposed for finding a super-stable matching, or reporting that none exists, in various two-sided matching markets \cite{Irving94,manlove99,irving-manlove-scott-hospitals,RCIIL-ec14}.  

In the next section we provide definitions of notation and terminology, leading to formal statements of the problems under consideration in this paper.  A roadmap of the remaining sections is then given at the end of Section \ref{sec:defs}.

%
\section{Preliminary definitions and results}
\label{sec:defs}
\subsection{\smpi, \smti, and levels of stability}
In an instance of the \emph{Stable Marriage problem with Partially ordered preferences and Incomplete lists (\smpi)}, there are two sets of agents, namely a set of men $\manset=\{\man_1, \man_2, \ldots, \\ \man_{\numofmen}\}$, and a set of women $\womanset=\{\woman_1, \woman_2, \ldots, \woman_{\numofwomen}\}$.  We assume without loss of generality that $|\manset|=|\womanset|$ (we can easily reduce the case where the two sets are of different sizes to our setting). Let $[i]$ denote the set $\{1,2,\ldots,i\}$. We use the term \emph{agents} when making statements that apply to both men and women, and the term \emph{candidates} to refer to agents on the opposite side of the market to that of an agent under consideration. Each agent $\agent$ finds a subset of candidates acceptable -- we refer to these as $\agent$'s \emph{acceptable} candidates.  An agent $a$'s preferences over his/her acceptable candidates need not be strict.  That is, given two candidates, $a$ might not be able to compare them against each other.
We denote by $\partialorder_{\man_i}$ and $\partialorder_{\woman_j}$ the partial orders that represent the preferences of $\man_i$ and $\woman_j$, respectively. We let $\partialorder_{\manset,\womanset} =(\partialorder_{\man_1},\ldots,\partialorder_{\man_{\numofmen}},\partialorder_{\woman_1} ,\ldots,\partialorder_{\woman_{\numofwomen}})$ and call $\partialorder_{\manset,\womanset}$ a \emph{partial preference ordering profile}.  

Let $\inputsetting=(\manset,\womanset,\partialorder_{\manset,\womanset})$ be an instance of \smpi, let $\agent$ be an agent and let $\candidate_1$ and $\candidate_2$ be two acceptable candidates for $\agent$ in $I$.  We say that $\agent$ \emph{strictly prefers} $\candidate_1$ to $\candidate_2$ if $(\candidate_1,\candidate_2) \in \partialorder_{\agent}$, and we say that $\agent$ \emph{cannot compare} $\candidate_1$ and $\candidate_2$ (or that $\agent$ finds $\candidate_1$ and $\candidate_2$ \emph{incomparable}) if $(\candidate_1,\candidate_2)\notin \partialorder_{\agent}$ and $(\candidate_2,\candidate_1)\notin \partialorder_{\agent}$.
We sometimes use the graph-theoretic representation of $\partialorder_{\agent}$ where candidates in $\partialorder_{\agent}$ correspond to vertices and there is an arc from a candidate $\candidate_i$ to a candidate $\candidate_j$ if and only if $(\candidate_i,\candidate_j) \in \partialorder_{\agent}$.

An instance $\inputsetting'=(\manset,\womanset,\partialorder'_{\manset,\womanset})$ of \smpi\ is a \emph{refinement} of $\inputsetting$ if for each agent $\agent$, any strict total order that is a linear extension of $\partialorder'_{\agent}$ is also a linear extension of $\partialorder_{\agent}$.  We may also refer to $\partialorder'_{\manset,\womanset}$ being a refinement of $\partialorder_{\manset,\womanset}$ (or indeed $\inputsetting$) using the same definition.  Also we can define $\partialorder'_\agent$ being a refinement of $\partialorder_\agent$ for some specific agent $\agent$ similarly.
\begin{observation}\label{obs:refinement}
Given two instances $\inputsetting$ and $\inputsetting'$ of \smpi, $\inputsetting'$ is a refinement of $\inputsetting$ if and only if the following condition holds: for each agent $\agent$ and every two candidates $\candidate_1$ and $\candidate_2$ acceptable to $\agent$, if $(\candidate_1,\candidate_2)\in\partialorder_{\agent}$ then $(\candidate_1,\candidate_2)\in\partialorder'_{\agent}$.
\end{observation}

A well studied special case of \smpi\ is the \emph{Stable Marriage problem with Ties and Incomplete lists (\smti)}, in which incomparability is transitive and is interpreted as indifference.
In \smti, each agent has a partition of acceptable candidates into \emph{indifference classes} or \emph{ties} such that he or she is indifferent between the candidates in the same indifference class, but has a strict preference ordering over the indifference classes. In an instance of \smti, let $\eqclass^{\agent}_t$ denote the $t$-th indifference class of agent $\agent$, where $t \in [\numofwomen]$. We assume that $\eqclass^{\agent}_t = \emptyset$ implies $\eqclass^{\agent}_{t'} = \emptyset$ for all $ t' >t$. The \emph{Stable Marriage problem with Incomplete lists (\smi)} is the special case of \smti\ in which each tie is of size one.  Similarly the \emph{Stable Marriage problem with Ties (\smt)} is the special case of \smti\ in which each man finds each woman acceptable and vice versa.  

Given an instance $I$ of \smpi, a \emph{matching} $\matching$ is a pairing of men and women such that each man is paired with at most one woman and vice versa, and no agent is matched to an unacceptable partner. If $\man$ and $\woman$ are matched in $\matching$ then $\matching(\man)=\woman$ and $\matching(\woman)=\man$. We say that $\matching(\agent) = \varnothing$ if $\agent$ is unmatched under $\matching$. Different levels of stability can be defined in the context of \smpi\ \cite{Irving94,manlove99}. A \emph{strong blocking pair} is an acceptable (man,woman) pair, each of whom is unmatched or strictly prefers the other to his/her partner.  A \emph{weakly stable matching} is a matching with no strong blocking pair. Every instance of \smpi\ admits a weakly stable matching \cite{Manlove-etal}.  An acceptable (man,woman) pair is a \emph{weak blocking pair} if each member of the pair is either unmatched or strictly prefers the other to his/her partner or cannot compare the other with his/her partner, and one member of the pair is either unmatched or strictly prefers the other to his/her partner.  A \emph{strongly stable matching} is a matching with no weak blocking pair.  Finally a \emph{very weak blocking pair} is an acceptable (man,woman) pair, each of whom is unmatched or strictly prefers to other to his/her partner or cannot compare the other with his/her partner. A \emph{super-stable matching} is a matching with no very weak blocking pair. It can be easily verified that a matching is super-stable if and only if it is weakly stable w.r.t.\ all strict total orders that are linear extensions of 
the given partial preference orderings \cite[Lemma 3.2.4]{manlove-book}.  In instances of \smi, weak stability, strong stability and super-stability are all equivalent to classical stability.


\subsection{Interviews to refine the partial orders}
In a given instance $\inputsetting=(\manset,\womanset,\partialorder_{\manset,\womanset})$ of \smpi\ in this paper, we assume that $\partialorder_{\manset,\womanset}$, the partial preference ordering profile, represents the agents' initial information state.  That is, agents may not have enough information initially in order to rank their acceptable candidates in strict order.  However in the problem instances that we will later define in this section, we will assume that each agent $\agent$ has a strict preference ordering $\pref_{\agent}$ over his or her acceptable candidates.  This represents the true (and strict) underlying preferences over $\agent$'s acceptable candidates, although crucially, $\agent$ may not (and in general will not) initially be aware of the entire ordering.  We let $\pref_{\manset,\womanset} =(\pref_{\man_1},\ldots,\pref_{\man_{\numofmen}},\pref_{\woman_1} ,\ldots,\pref_{\woman_{\numofwomen}})$ and call $\pref_{\manset,\womanset}$ a \emph{strict (true underlying) preference ordering profile}.  The task of the agents is to learn enough information about their acceptable candidates in order to refine their preferences, in a manner consistent with $\pref_{\manset,\womanset}$, to obtain an \smpi\ instance $I'$ that admits a super-stable matching $\matching$ (thus $\matching$ will be stable with respect to $\pref_{\manset,\womanset}$).

Following the model introduced in \cite{RCIL-ec13}, we assume that instances can be refined through \emph{interviews}. Each interview pairs one man $\man$ with one woman $\woman$.
An interview is informative to both parties involved. Hence saying ``$\man$ interviews $\woman$'' is equivalent to saying ``$\woman$ interviews $\man$''. When agent $\agent$ interviews $\ell$ candidates, this results in a new refined \smpi\ instance which is exactly the same as $\inputsetting$ except that $\agent$ now has a strict preference ordering over the $\ell$ interviewed candidates.

Notice that if an agent interviews only one candidate, no refinement takes place. Note also that not all refinements of $\inputsetting$ can be reached by a set of interviews. For example, suppose that in $\inputsetting$ we have one man $\man_1$ and three women $\woman_1$, $\woman_2$, and $\woman_3$. Suppose $\man_1$ finds the three women acceptable and incomparable. Assume that in $\inputsetting'$ man $\man_1$ prefers $\woman_1$ to both $\woman_2$ and $\woman_3$, and cannot compare $\woman_2$ and $\woman_3$. It is easy to see that $\inputsetting'$ is a refinement of $\inputsetting$, but no set of interviews can reach $\inputsetting'$: for $\man_1$ to learn that he prefers $\woman_1$ to the other two women he must interview all three women, but then he will have a strict preference ordering over the three of them. 

We say that an \smpi\ instance $\inputsetting'$ is an \emph{interview-compatible refinement} of an \smpi\ instance $\inputsetting$ if $\inputsetting'$ can be refined from $\inputsetting$ using interviews.
We now show that interview-compatible refinements can be recognized easily.
\begin{proposition}
Let $\inputsetting$ and $\inputsetting'$ be two instances of \smpi.  We can determine in $O(n^3)$ time whether $\inputsetting'$ is an interview-compatible refinement of $\inputsetting$.
\label{prop1}
\end{proposition}
\begin{proof}
To verify whether $\inputsetting'$ is a refinement of $\inputsetting$, it is sufficient to check whether the condition of Observation~\ref{obs:refinement} holds. With a suitable data structure, we can do this in $O(n^3)$ time. 
For each agent $\agent$ identify the edges present in $\partialorder'_{\agent}$ that are not in $\partialorder_{\agent}$, and let $S(\agent)$ be the set of candidates in $\partialorder'_{\agent}$ that form an endpoint of at least one such edge.
For $\inputsetting'$ to be an interview-compatible refinement of $\inputsetting$, it is necessary and sufficient that, for every $\agent$, $S(\agent)$ forms a complete subgraph in the undirected graph corresponding to $\partialorder'_{\agent}$. This can be tested in $O(n^3)$ time overall.
\end{proof}

Let $\inputsetting'$ be an \smpi\ instance that is an interview-compatible refinement of a given \smpi\ instance $\inputsetting$. We define the \emph{cost} of $\inputsetting'$ given $\inputsetting$ to be the minimum number of interviews required to refine $\inputsetting$ into $\inputsetting'$.  The following proposition shows how to compute this cost efficiently.
\begin{proposition}
Let $\inputsetting$ be an \smpi\ instance and let $\inputsetting'$ be an interview-compatible refinement of $\inputsetting$.  We can determine in $O(n^3)$ time the cost of $\inputsetting'$ given $\inputsetting$.
\label{prop2}
\end{proposition}
\begin{proof}
We identify the set of interviews $\interviewset$ that refines $\inputsetting$ into $\inputsetting'$ as follows. Initially $\interviewset=\emptyset$.
For each agent $\agent$ and every two candidates $\candidate_1$ and $\candidate_2$, 
if $\agent$ cannot compare $\candidate_1$ and $\candidate_2$ under $\inputsetting$, but prefers one to the other under $\inputsetting'$, $\agent$ must have interviewed both $\candidate_1$ and $\candidate_2$. Add both of these interviews to $\interviewset$. Notice that we might have already accounted for one or both of these interviews. However since $\interviewset$ is a set, no interview is going to be included in $\interviewset$ more than once. With a suitable data structure, the aforementioned procedure terminates in $O(\numofmen^3)$ time overall, and once it does, $|\interviewset|$ denotes the cost of $\inputsetting'$.
\end{proof}
%


\subsection{Definition of interview minimizing problems}
The motivating problem is as follows: given an instance $\inputsetting=(\manset,\womanset,\partialorder_{\manset,\womanset})$ of \smpi, find an interview-compatible refinement $\inputsetting'$ of minimum cost such that $\inputsetting'$ admits a super-stable matching.  Since the result of one interview might influence which interview/s to carry out next, any strategy for carrying out interviews should be regarded as an online algorithm.  

In fact there may be no online algorithm that can guarantee to provide an optimal solution in all cases.  To see this, let us return to the example of Section \ref{sec:intro} involving two residents and two hospitals, and suppose that initially everyone finds the two agents on the other side of the market acceptable and incomparable.  Let the true underlying preferences be given by Figure \ref{fig:example} (here, preference lists are ordered from left to right in decreasing order of preference).
\begin{figure}[t]
\[
\begin{array}{rlrl}
r_1: & h_1 ~ h_2 ~~~~~~ & h_1: & r_2 ~ r_1\\
r_2: & h_2 ~ h_1 ~~~~~~ & h_2: & r_2 ~ r_1
\end{array}
\]
\caption{A stable marriage problem instance.}
\label{fig:example}
\end{figure}
Any online interviewing strategy must start with one interview in the absence of any knowledge; w.l.o.g.\ suppose that the first interview involves $r_1$ and $h_1$.  Then it may be verified that the algorithm is bound to use 3 more interviews before a super-stable matching can be found. (If one interview does not take place then two agents on opposite sides of the market cannot compare the two candidates in their preference list; it then follows that each of the two possible matchings would be blocked according to super-stability.)
On the other hand the interview involving $r_1$ and $h_1$ was unnecessary and an optimal strategy uses only 3 interviews.

Towards computing bounds for the competitive ratio of an online algorithm, the offline scenario is of interest, and that is what we consider in what follows.  In the offline case, the mechanism designer is given $\pref_{\manset,\womanset}$, the strict (true underlying) preference ordering profile of the agents,
and would like to compute an optimal interviewing schedule, i.e., an interview-compatible refinement $\inputsetting'$ of $\inputsetting$, such that $\pref_{\manset,\womanset}$ refines $\inputsetting'$. This is reflected in the definition of the following problem, named \minicr, which is an abbreviation for ``Minimum-cost Interview Compatible Refinement problem''.

\begin{definition}\label{def:minicr}
An instance of \minicr\ comprises a tuple $(I,\pref_{\manset,\womanset})$, where $I$ is an instance of \smpi\ and $\pref_{\manset,\womanset}$ is a strict preference ordering profile that refines $I$.  The problem is to find an interview-compatible refinement $\inputsetting'$ of $\inputsetting$ such that (i) $\pref_{\manset,\womanset}$ refines $\inputsetting'$, (ii) $\inputsetting'$ admits a super-stable matching, and (iii) $\inputsetting'$ is of minimum cost amongst interview-compatible refinements that satisfy (i) and (ii). 
\end{definition}

The decision version of \minicr\ is defined as follows.
\begin{definition}\label{def:icr-dec}
An instance of \icr\ comprises a tuple $(I,\pref_{\manset,\womanset},K)$, where $I$ is an instance of \smpi, $\pref_{\manset,\womanset}$ is a strict preference ordering profile that refines $I$, and $K$ is a non-negative integer.  The problem is to decide whether there exists an interview-compatible refinement $\inputsetting'$ of $\inputsetting$, with cost at most $K$, such that $\pref_{\manset,\womanset}$ refines $\inputsetting'$ and $\inputsetting'$ admits a super-stable matching.
\end{definition}

As discussed in Section \ref{sec:intro}, it is sometimes the case that we aim for a particular matching, stable under $\pref_{\manset,\womanset}$, that has some desirable properties, for example the woman-optimal stable matching.  The offline problem can then be
viewed as a restricted variant of \minicr\ where, in addition to $\inputsetting$ and $\pref_{\manset,\womanset}$, we are also equipped with a matching $\matching$. This is reflected in the definition of the following problem, named \minicrx, which is an abbreviation for ``Minimum-cost Interview Compatible Refinement problem with Exact matching''.

\begin{definition}\label{def:minicr-exact}
An instance of \minicrx\ comprises a tuple $(\inputsetting,\pref_{\manset,\womanset},\matching)$, where $\inputsetting$ is an instance of \smpi, $\pref_{\manset,\womanset}$ is a strict preference ordering profile that refines $\inputsetting$, and $\matching$ is a matching that is weakly stable w.r.t.\ $\pref_{\manset,\womanset}$. The problem is to find an interview-compatible refinement $\inputsetting'$ of $\inputsetting$, such that (i) $\pref_{\manset,\womanset}$ refines $\inputsetting'$, (ii) $\matching$ is super-stable in $\inputsetting'$, and (iii) $\inputsetting'$ is of minimum cost amongst interview-compatible refinements of $\inputsetting$ that satisfy (i) and (ii).
\end{definition}

The decision version of \minicrx, called \icrx, is then defined analogously to the way that \icr\ was obtained from the definition of \minicr.

The remainder of this paper is organized as follows.  In Section \ref{sec:npc} we first show that \icr\ is NP-complete even under quite restricted settings. The proof is by reduction from Vertex Cover. We also leverage the same proof to show that \icrx\ is also NP-complete.
Then in Section \ref{sec:poly} we provide a reverse reduction, from Vertex Cover to \icrx, and utilize it to show that \minicrx\ is polynomial-time solvable for several restricted settings.  Some concluding remarks are presented in Section \ref{sec:conc}.

\section{NP-completeness results}
\label{sec:npc}
We show that \icr\ is NP-complete even if $\inputsetting$ is an instance of \smti\ in which each indifference class is of size at most 3. Further, we prove that \icr\ is NP-complete even for \smt\ instances, and even if all men are indifferent between all women.  We first provide a lemma that will come in handy in proving our claims; the proof is straightforward, and is omitted.
\hide{
\begin{lemma}\label{lem:is-po}
We can decide in polynomial time whether a given $\partialorder_{\manset,\womanset}$ is a partial order.
\end{lemma}
\begin{proof}
We need to verify whether $\partialorder_{\manset,\womanset}$ is irreflexive, asymmetric, and transitive. The first two can be verified easily in $O(n^2)$ by checking each agent. Verifying transitivity can be done in $O(n^4)$ time, by e.g., using an algorithm to compute the transitive closure of $\partialorder_{\manset,\womanset}$. $\partialorder_{\manset,\womanset}$ is transitive if and only if it is the same as its transitive closure.
\end{proof}
}
\begin{lemma}\label{lem:digraph}
Let $G=(V,E)$ be an undirected graph where for each vertex $v$, $deg(v) \leq 3$. We can direct the edges in $E$ such that for each $v$, $deg^+(v)\leq 2$ and $deg^-(v)\leq 2$.
\end{lemma}

Unlike many problems that are NP-complete, the membership of \icr\ in the class NP is not trivial. Hence, we provide a proof via the following lemma.
\begin{lemma}\label{ICR-in-np}
\icr\ is in NP.
\end{lemma}
\begin{proof}
To prove this, it is sufficient to show that given \smpi\ instances $\inputsetting$ and $\inputsetting'$, a strict preference profile $\pref_{\manset,\womanset}$ and an integer $K$, we can decide in polynomial time whether (i) $\inputsetting'$ is an interview-compatible refinement of $\inputsetting$, (ii) $\inputsetting'$ has cost at most $K$, (iii) $\pref_{\manset,\womanset}$ refines $\inputsetting'$, and (iv) $\inputsetting'$ admits a super-stable matching. 

Both (i) and (ii) are established by Propositions \ref{prop1} and \ref{prop2} respectively.
For (iii), it is straightforward to check in polynomial time whether $\pref_{\manset,\womanset}$ refines $\inputsetting'$. 

Finally to establish (iv), we can use the polynomial time algorithm of \cite{RCIIL-ec14}, SUPER-SMP, to decide whether $\inputsetting'$ admits a super-stable matching or not.
\end{proof}

We show that \icr\ is NP-hard by reducing from the decision version of the Vertex Cover problem (\vc). \vc\ is defined as follows: given a graph $G=(V,E)$ and an integer $K$, decide whether $G$ admits a vertex cover of size at most $K$. \vc\ is NP-complete even if each vertex has degree at most $3$ \cite{GJ77}; let \vct\ denote this restriction. We denote by \minvc\ the optimization version of \vc, that is the problem of finding a minimum vertex cover in a given graph $G$.

\begin{theorem}\label{ICR-np-hard}
\icr\ is NP-complete even for \smti\ instances in which each indifference class has size \mbox{at most $3$}.
\end{theorem}
\begin{proof}
By Lemma \ref{ICR-in-np}, \icr\ belongs to NP.  To show NP-hardness, we reduce from \vct. Let the undirected simple graph $G=(V,E)$ be given such that $deg(v)\leq 3$, $\forall v\in V$. Let $V=\{v_1,\ldots,v_n\}$. Let $G'=(V,E')$ be a digraph where (i) $\forall (v_i,v_j) \in E'$, $(v_i,v_j) \in E$, (ii) $\forall (v_i,v_j) \in E$, either $(v_i,v_j) \in E'$ or  $(v_j,v_i) \in E'$ (but not both), and (iii) $\forall v \in V$, $deg^+(v)\leq 2$ and $deg^-(v)\leq 2$. Note that by Lemma~\ref{lem:digraph} such a graph $G'$ exists. 
We create an instance $\inputsetting=(\manset,\womanset,\partialorder_{\manset,\womanset})$ of \smti\ as follows:
\begin{itemize}
	\item For each vertex $v_i \in V$ there is a man $\man_i \in \manset$ and a woman $\woman_i \in \womanset$. That is $\manset=\{\man_i | \forall v_i\in V\}$ and $\womanset=\{\woman_i | \forall v_i\in V\}$.
	\item Each man $\man_i$ finds acceptable $\woman_i$ and all women $\woman_j$ such that  $(v_i,v_j)\in E'$. Moreover, each man $\man_i$ is indifferent between all his acceptable women. 
	\item Each woman $\woman_i$ finds acceptable $\man_i$ and all men $\man_j$ such that  $(v_j,v_i)\in E'$. Moreover, each woman $\woman_i$ is indifferent between all her acceptable men.
\end{itemize}

Note that the total length of the men's preference lists is equal to $|E|+|V|$. Also note that as $deg^+(v)\leq 2$ and $deg^-(v)\leq 2$ for all $v \in V$, hence the indifference classes are of size at most $3$.
%
Let $\pref_{\manset,\womanset}$ be a strict preference ordering under which $\man_i$ and $\woman_i$ rank each other at the top of their preference lists. We prove that $G$ has a vertex cover of size at most $K$ if and only if there exists a refinement $\inputsetting'$ of $\inputsetting$, of cost at most $K' = K+|E|$, such that $\pref_{\manset,\womanset}$ refines $\inputsetting'$ and $\inputsetting'$ admits a super-stable matching. Notice that $\pref_{\manset,\womanset}$ admits only one stable matching, that being $\matching=\{(\man_i,\woman_i)|\forall i\}$. This implies that $\inputsetting'$ must admit exactly one super-stable matching, that being $\matching$. 


\emph{Proof of the only-if direction: Assume that $G$ has a vertex cover $C$ of size $k \leq K$. We show that there is a refinement $\inputsetting'$ of cost $k'=k+|E| \leq K'$ such that $\pref_{\manset,\womanset}$ refines $\inputsetting'$ and $\matching$ is super-stable in $\inputsetting'$.} We create $\inputsetting'$ as follows. For each vertex $v_i\in C$ refine $\inputsetting$ such that both $\man_i$ and $\woman_i$ now have strict preferences as in $\pref_{\manset,\womanset}$
This refinement can of course be achieved by having both $\man_i$ and $\woman_i$ interview all candidates in their lists; this includes an interview between $\man_i$ and $\woman_i$. 
Notice that since interviews are informative to both parties involved, partial refinements in the preference orderings of those persons whose corresponding vertices are not in $C$ must have taken place as well. For example, consider a case in which $v_i,v_j \in C$, $v_k \notin C$, and $(v_i,v_k), (v_j,v_k) \in E'$. Then both $\man_i$ and $\man_j$ interview $\woman_k$ and hence $\woman_k$ must now rank $\man_i$ and $\man_j$ in strict order of preferences as in $\pref_{\manset,\womanset}$. 
An interview is either between same indexed agents, e.g., between $\man_i$ and $\woman_i$, or between agents of different indices, e.g., between $\man_i$ and $\woman_j$ where $i\neq j$. We refer to an interview of the former type as a \emph{same-index} interview and an interview of the latter type as a \emph{different-index} interview. The total number of interviews performed by all agents is going to be $k$ same-index interview plus a number of different-index interviews. The number of different-index interviews under any refinement can be at most $|E|$, and under our proposed refinement is exactly $|E|$ (since $C$ is a vertex cover). Therefore the total number of interviews is exactly $k+|E|$.
It remains to show that $\matching$ is a super-stable matching in $\inputsetting'$. We call a $(man,woman)$ pair a \emph{fixed pair} if they are matched in every stable matching of every strict order refinement of $\inputsetting'$.  We show that $(\man_i,\woman_i)$ is a fixed pair for all $i\leq n$, hence proving that $\matching$ is the only stable matching in every strict order refinement of $\inputsetting'$ and therefore definitely a super-stable matching in $\inputsetting'$. Take any pair $(\man_i,\woman_i)$ such that $v_i \in C$. By our construction, $\man_i$ and $\woman_i$ rank each other at top, so clearly $(\man_i,\woman_i)$ is a fixed pair.
Now take any pair $(\man_j,\woman_j)$ such that $v_j \notin C$. Since $v_j$ is not in the vertex cover, therefore it must be the case that all neighbors of $v_j$ are in $C$. Thus, for any $v_k$ that is a neighbor of $v_j$, it has been already established that $(\man_k,\woman_k)$ is a fixed pair. Therefore $(\man_j,\woman_j)$ is also a fixed pair.  Moreover, neither $\man_j$ nor $\woman_j$ can form a very weak blocking pair with a person they are not matched to. 

\emph{Proof of the if direction: Assume that $\inputsetting$ has a refinement $\inputsetting'$ of cost $k' \leq K'$ such that $\pref_{\manset,\womanset}$ refines $\inputsetting'$ and $\matching$ is super-stable in $\inputsetting'$. We show that $G$ admits a vertex cover of size at most $k'-|E| \leq K$.} We first show that in order to arrive at $\inputsetting'$, every agent should have interviewed every candidate s/he finds acceptable and to whom s/he is not matched. Assume for a contradiction that this is not the case. Take a pair $(\man_i,\woman_j)$, acceptable to each other, who have not interviewed. Therefore, under $\inputsetting'$, $\man_i$ is indifferent between $\woman_j$ and $\woman_i$ (to whom he is matched in $\matching$), and $\woman_j$ is indifferent between $\man_i$ and $\man_j$ (to whom she is matched in $\matching$). Hence $(\man_i,\woman_j)$ constitutes a very weak blocking pair in $\matching$ under $\inputsetting'$, a contradiction. We have established so far that every agent must have interviewed acceptable candidates to whom s/he is not matched, which means that each agent has interviewed all candidates in his/her list who have a different index from him/her. This amounts to the total of $|E|$ interviews. The only remaining interviews for which we have not yet accounted are those corresponding to matched pairs. Let $C$ be a set of vertices such that vertex $v_i$ is in $C$ if and only if $\man_i$ and $\woman_i$ have interviewed under $\inputsetting'$. Notice that $|C| = k' - |E|$. Take any $v_j \notin C$. We show that all neighbors of $v_j$ are in $C$, establishing that $C$ is a vertex cover. Since $v_j \not\in C$, it follows from the construction of $C$ that $\man_j$ and $\woman_j$ have not interviewed under $\inputsetting'$. Assume for a contradiction that $v_j$ has a neighbor, say $v_k$, who too is not in $C$. Therefore $\man_k$ and $\woman_k$ have not interviewed under $\inputsetting'$ either. W.l.o.g. assume that $(v_j,v_k)\ \in E'$. (A similar argument applies if $(v_k,v_j) \in E'$.) Therefore $\man_j$ and $\woman_k$ are acceptable to each other. Furthermore, since neither $\man_j$ nor $\woman_k$ have interviewed their partners in $\matching$, it is the case that $\man_j$ is indifferent between $\woman_j$ (his partner in $\matching$) and $\woman_k$, and $\woman_k$ is indifferent between $\man_k$ (her partner in $\matching$) and $\man_j$. Therefore $(\man_j,\woman_k)$ constitutes a very weak blocking pair in $\matching$ under $\inputsetting'$, a contradiction. 
\end{proof}
\hide{
The next result then follows directly from Theorem~\ref{ICR-in-np} and Theorem~\ref{ICR-np-hard}.
\begin{corollary}
\icr\ is NP-complete even for \smti\ instances and even when each indifference class is of size at most $3$.
\end{corollary}
}

We next show that \icr\ is also NP-complete under a different restricted setting by making small alterations to the proof of Theorem \ref{ICR-np-hard}.

\begin{corollary}\label{cor:np-complete-smt}
\icr\ is NP-complete even for \smt\ instances and even if agents on one side of the market are indifferent between all the candidates.
\end{corollary}
\begin{proof}
W.l.o.g. assume that all men are indifferent between all women.
Modify the reduction presented in the proof of Theorem~\ref{ICR-np-hard} as follows.
\begin{itemize}
	\item For each vertex $v_i \in V$ there is a man $\man_i$ in $\manset$ and a woman $\woman_i$ in $\womanset$. That is $\manset=\{\man_i | \forall v_i\in V\}$ and $\womanset=\{\woman_i | \forall v_i\in V\}$.
	\item Each man $\man_i$ finds all women acceptable and is indifferent between them.
	\item Each woman $\woman_i$ finds all men acceptable and has two indifference classes. In the top indifference class are $\man_i$ and all men $\man_j$ such that $(v_i,v_j)\in E$. In the second indifference class are all other men.
\end{itemize}

Note that the total length of the women's first indifference classes is equal to $2|E| + |V|$.  
Let $\pref_{\manset,\womanset}$ be a strict preference ordering under which $\man_i$ and $\woman_i$ rank each other at the top of their preference lists. We prove that $G$ has a vertex cover of size at most $K$ if and only if there exists a refinement $\inputsetting'$ of $\inputsetting$, of cost at most $K' = K+2|E|$, such that $\pref_{\manset,\womanset}$ refines $\inputsetting'$ and $\inputsetting'$ admits a super-stable matching. Notice that $\pref_{\manset,\womanset}$ admits only one stable matching, that being $\matching=\{(\man_i,\woman_i)|\forall i\}$. This implies that $\inputsetting'$ must admit exactly one super-stable matching, that being $\matching$. Modify the proof of Theorem~\ref{ICR-np-hard} as follows.
 
\emph{In the only-if direction}: For each vertex $v_i \in C$ refine $\inputsetting$ such that $\man_i$ has a strict preference ordering, as in $\pref_{\man_i}$, over women in $\{\woman_i\}\cup\{\woman_j| (v_i,v_j)\in E \}$ and $\woman_i$ has a strict preference ordering, as in $\pref_{\woman_i}$, over men in $\{\man_i\}\cup\{\man_j| (v_i,v_j)\in E \}$. Consequently, for all $v_j$ adjacent to $v_i$, $\man_i$ prefers $\woman_i$ to $\woman_j$ and likewise $\woman_i$ prefers $\man_i$ to $\man_j$. This refinement can be achieved by having $\man_i$ interview $\woman_i$ and all $\woman_j$ such that $(v_i,v_j)\in E$, and additionally having $\woman_i$ interview all $\man_j$ such that $(v_i,v_j) \in E$. The number of different-index interviews under any refinement can be at most $2|E|$, and under our proposed refinement is exactly $2|E|$ (since $C$ is a vertex cover). So the total number of interviews is exactly $k+2|E|$. It remains to show that $\matching$ is a super-stable matching in $\inputsetting'$. Assume for a contradiction that there exists a very weak blocking pair $(\man_i,\woman_j)$. 
	\begin{itemize} 
		\item If $(v_i,v_j) \in E$, then $v_i$ or $v_j$ is in $C$. If $v_i \in C$ then $\man_i$ and $\woman_i$ have interviewed and therefore $\woman_i \pref_{\man_i} \woman_j$. If $v_j \in C$ then $\man_j$ and $\woman_j$ have interviewed and therefore $\man_j \pref_{\woman_j} \man_i$. Both cases imply that $(\man_i,\woman_j)$ is not a very weak blocking pair, a contradiction.
		\item If $(v_i,v_j) \notin E$ then $\man_j \pref_{\woman_j} \man_i$, therefore $(\man_i,\woman_j)$ is not a very weak blocking pair, a contradiction.
	\end{itemize}
	
	\emph{In the if direction}: We show that in order to arrive at $\inputsetting'$, every man $\man_i$ should have interviewed all women $\woman_j$ such that $(v_i,v_j) \in E$, and likewise every woman $\woman_i$ should have interviewed all men $\man_j$ such that $(v_i,v_j) \in E$. The proof is similar to that presented in the proof of Theorem \ref{ICR-np-hard}.
	Hence we can conclude that at least $2|E|$ different-index interviews must have taken place in the refinement. The rest of the proof is similar to that presented in the proof of Theorem \ref{ICR-np-hard}, with the difference that $|C| \leq k' - 2|E|$.
\end{proof}

In the proof of Theorem~\ref{ICR-np-hard}, $\matching$ is the unique stable matching under $\pref_{\manset,\womanset}$. Therefore, 
it follows from the proofs of Theorem~\ref{ICR-np-hard} and Corollary~\ref{cor:np-complete-smt} that \icrx\ is also NP-complete for the restrictions stated in those results. 
\begin{corollary}\label{cor:icrx}
\icrx\ is NP-complete even for \smti\ instances, and even when each indifference class is of size at most $3$. \icrx\ is also NP-complete even for \smti\ instances and even if agents on one side of the market are indifferent between all the candidates.
\end{corollary}

We remark that Theorem 4.4 of~\cite{RCIL-ec13} implies that \icrx\ is NP-complete and, likewise, Corollary~\ref{cor:icrx} implies Theorem 4.4 of~\cite{RCIL-ec13}. However, Corollary~\ref{cor:icrx} is stronger as it is stated for a more restricted setting. 


\section{Polynomial-time solvable variants}
\label{sec:poly}
\subsection{Preliminaries}
In this section we explore the tractability of \minicrx\ under various restricted settings. Recall that we have reduced from \vc\ to \icr\ and \icrx\ in order to show that these problems are NP-hard. Here we present a reverse reduction, from \icrx\ to \vc, that will come in handy in proving our claims.


Let an instance $(\inputsetting,\pref_{\manset,\womanset},\matching)$ of \icrx\ be given. As $\matching$ is weakly stable w.r.t.\ $\inputsetting$, it admits no strong blocking pair.  If $\matching$ is not super-stable w.r.t.\ $\inputsetting$, then $\matching$ must admit some very weak blocking pairs. We refer to such blocking pairs as \emph{potential blocking pairs}. We distinguish between between potential blocking pairs by the degree of choice one has when attempting to resolve them.

\begin{definition} [Potential Blocking Pair (PBP)]
Given an \icrx\ instance $(\inputsetting,\pref_{\manset,\womanset},\matching)$, a pair $(\man,\woman)$ is a \emph{potential blocking pair (PBP)} if $(\man,\woman)$ is a very weak blocking pair under $\inputsetting$. Each PBP $(\man,\woman)$ belongs to either of the following classes.
\begin{itemize} 
	\item \emph{Potential Blocking Pair of Degree 1 (PBP-D1)} if either $\man$ or $\woman$ strictly prefers the other to his or her current partner under $\pref_{\manset,\womanset}$. 
	\item \emph{Potential Blocking Pair of Degree 2 (PBP-D2)} if both $\man$ and $\woman$ strictly prefer their partners to each other under $\pref_{\manset,\womanset}$. 
\end{itemize}
\end{definition}
Let $\inputsetting'$ be an interview-compatible refinement of $\inputsetting$. We say that a given potential blocking pair of $\inputsetting$, $(\man,\woman)$, is \emph{resolved} under $\inputsetting'$ if $(\matching(\man),\woman)\in \partialorder'_{\man}$ or $(\matching(\woman),\man) \in \partialorder'_{\woman}$.

If $(\man,\woman)$ is a PBP-D2, then it must be that $\man$ and $\woman$ cannot compare each other and their current partners under $\inputsetting$.  Thus in order to resolve $(\man,\woman)$ it is sufficient, and necessary, that $\man$ or $\woman$ learn his/her true preference ordering over his/her partner and the other side. 

Let $(\man,\woman)$ be a PBP-D1 and assume that $\man$ strictly prefers $\woman$ to $\matching(\man)$ (the argument is similar if $\man \pref_{\woman} \matching(\woman)$). Therefore, $\woman$ must find $\man$ and $\matching(\woman)$ incomparable under $\inputsetting$, or $(\man,\woman)$ either blocks $\matching$ or is not a PBP, and $\matching(\woman) \pref_{\woman} \man$, or $(\man,\woman)$  blocks $\matching$. Furthermore, in order to resolve this PBP $\woman$ has to learn that she prefers $\matching(\woman)$ to $\man$.

In what follows we use $PBP$, $PBP_1$, and $PBP_2$ to refer to the set of potential blocking pairs, and those of degree 1 and degree 2 respectively. 


\begin{proposition}\label{prop:resolved}
Let $(\inputsetting,\pref_{\manset,\womanset},\matching)$ be an instance of \icrx\ and $\inputsetting'$ be an interview-compatible refinement of $\inputsetting$. Then $\matching$ is super-stable under $\inputsetting'$ if and only if all PBPs in $\inputsetting$ are resolved under $\inputsetting'$. 
\end{proposition}

It is easy to see that for a potential blocking pair $(\man,\woman)$ to be resolved, at least one of $\man$ or $\woman$  needs to interview both the other side and his or her current partner and conclude that s/he prefers his or her current partner to the other side. The next proposition then immediately follows.

\begin{proposition}\label{prop:pbp}
Let $(\inputsetting,\pref_{\manset,\womanset},\matching)$ be an instance of \icrx\ and $\inputsetting'$ be an interview-compatible refinement of $\inputsetting$. Then $\matching$ is super-stable under $\inputsetting'$ only if, for all $(\man,\woman) \in PBP$, $\man$ and $\woman$ have interviewed under $\inputsetting'$.
\end{proposition}


For each agent $\agent \in \manset\cup \womanset$ let $PBP_1(\agent)$ denote the set of candidates $\candidate$ such that either $(\agent,\candidate)$ or $(\candidate, \agent)$ is in $PBP_1$ and $\agent \pref_{\candidate} \matching(\candidate)$. 

\begin{lemma}\label{lem:type2}
Let $(\inputsetting,\pref_{\manset,\womanset},\matching)$ be an instance of \icrx\ and $\inputsetting'$ be an interview-compatible refinement of $\inputsetting$. Then $\matching$ is super-stable under $\inputsetting'$ only if $\agent$ has interviewed $\matching(\agent)$under $\inputsetting'$ for all agents $\agent$ where $PBP_1(\agent)\neq \emptyset$.
\end{lemma}
\begin{proof}
Assume for a contradiction that there exists an agent $\agent$ where $PBP_1(\agent)\neq \emptyset$ and $\agent$ has not interviewed $\matching(\agent)$. Therefore for every $\candidate \in PBP_1(\agent)$ it is still the case that $\agent$ cannot compare $\candidate$ and $\matching(\agent)$, and $\candidate$ prefers $\agent$ to $\matching(\candidate)$. Hence there exists at least one unresolved PBP under $\inputsetting'$. 
\end{proof}

\subsection{Reduction from \icrx\ to \vc}


Let $(\inputsetting,\pref_{\manset,\womanset},\matching)$ be an instance of \icrx. Let $\manset'=\{\man|PBP_1(\man)\neq \emptyset \vee PBP_1(\matching(\man))\neq \emptyset\}$.
Let $G(\inputsetting,\matching)=(V,E)$ be an undirected graph whose vertices $V$ correspond to matched pairs $(\man,\matching(\man))$. 
Let $PBP'_2 = \{(\man,\woman)| (\man,\woman) \in PBP_2, \man \notin \manset', \matching(\woman)\notin\manset'\}$. 
Let there be an edge between any two vertices $(\man,\matching(\man))$ and $(\man',\matching(\man'))$ if $(\man,\matching(\man')) \in PBP'_2$ or $(\man',\matching(\man)) \in PBP'_2$. Remove any vertex with degree zero. Note that for any remaining vertex $(\man,\matching(\man))$ it is the case that $\man \notin \manset'$.

\begin{theorem}\label{thm:reduction-to-VC}
$G(\inputsetting,\matching)$ has a vertex cover of size at most $K$ if and only if there exists a refinement $\inputsetting'$ of $\inputsetting$, of cost at most $K'=|PBP|+|\manset'|+K$, 
such that $\pref_{\man,\woman}$ refines $\inputsetting'$ and $\matching$ is super-stable in $\inputsetting'$. 
\end{theorem}
\begin{proof}
Assume that $G(\inputsetting,\matching)$ has a vertex cover $C$ of size at most $K$. Let $\inputsetting'$ be a refinement of $\inputsetting$ under which the following interviews have taken place. 
\begin{enumerate}
	\item Each pair $(\man,\woman) \in PBP$ interview each other -- a total of $|PBP|$ interviews.
	\item Each $\man \in \manset'$ interviews his partner $\matching(\man)$ -- a total of $|\manset'|$ interviews.
	\item Each pair $(\man,\matching(\man)) \in C$ interview each other -- a total of $K$ interviews.
\end{enumerate}
The total number of interviews is then equal to $|PBP|+|\manset'|+K$.
As a result of the above interviews, each agent $\agent$ learns his or her strict preference ordering over the interviewed candidates, as in $\pref_{\agent}$.
(Recall that the interviews are informative to both sides.) It is then easy to see that all PBP-D1's are resolved. It is also straightforward to see that for a PBP-D2 $(\man,\woman)$, if either $\man \in \manset'$ or $\matching(\woman) \in \manset'$, then $(\man,\woman)$ is resolved under $\inputsetting'$. It remains to show that the remaining PBP-D2's, that is those in $PBP'_2$, are resolved as well. Let $(\man,\woman)$ be such a PBP-D2. By the construction of $G(\inputsetting,\matching)$, $V$ includes $(\man,\matching(\man))$ and $(\matching(\woman), \woman)$ and there is an edge between these two vertices. As $C$ is a vertex cover, at least one of $(\man,\matching(\man))$ or $(\matching(\woman), \woman)$ belongs to $C$. If $(\man,\matching(\man)) \in C$ then, following the results of the interviews, $\man$ prefers $\matching(\man)$ to $\woman$ under $\inputsetting'$. (A similar argument holds for $\woman$ if $(\matching(\woman), \woman) \in C$.) Therefore $(\man,\woman)$ is resolved under $\inputsetting'$. 

Conversely, assume that $\inputsetting$ admits an interview-compatible refinement $\inputsetting'$ of size at most $K'$ such that $\matching$ is super-stable in $\inputsetting'$. We show that $G(\inputsetting,\matching)$ admits a vertex cover of size at most $K' - (|PBP|+|\manset'|)$. Let $C$ be a set of vertices $(\man,\matching(\man))$ in $V$ where $\man$ and $\matching(\man)$ have interviewed under $\inputsetting'$. Note that as we have removed all vertices of degree zero from $G(\inputsetting,\matching)$, hence all remaining vertices are adjacent to at least one edge corresponding to a member of $PBP'_2$.
We show that $C$ is a vertex cover and then prove an upper bound on the size of $C$. 

\textbf{$C$ is a vertex cover}: Let $((\man,\matching(\man)), (\man', \matching(\man')))$ be any edge in $E$. By the construction of $G(\inputsetting,\matching)$, $(\man,\matching(\man'))$ or $(\man',\matching(\man))$ is in $PBP'_2$. Assume that $(\man,\matching(\man')) \in PBP'_2$. (The argument for the case where $(\man',\matching(\man)) \in PBP'_2$ is similar.) As $(\man,\matching(\man'))$ is resolved under $\inputsetting'$, either $\man$ prefers his partner to $\matching(\man')$ under $\inputsetting'$, or $\matching(\man')$ prefers her partner to $\man$ under $\inputsetting'$. If the former, then $\man$ must have interviewed $\matching(\man)$ and hence $(\man,\matching(\man)) \in C$, and if the latter then $\matching(\man')$ must have interviewed $\man'$ and thus $(\man', \matching(\man')) \in C$ . Thus $C$ is a vertex cover. 

\textbf{$C$ is of size at most $K' - (|PBP|+|\manset'|)$}: We prove this by computing a lower bound on the number of interviews that do not correspond to a vertex in $C$. It follows Proposition~\ref{prop:pbp} that all PBPs must have interviewed, hence a total of $|PBP|$ interviews. It also follows Lemma~\ref{lem:type2} that each agent $\agent$ with $PBP_1(\agent)\neq \emptyset$ must have interviewed his/her partner.
Looking at this from men's perspective, all men $\man$ must interview $\matching(\man)$ if $PBP_1(\man)\neq \emptyset$ or
$PBP_1(\matching(\man))\neq \emptyset$ -- hence a total of $|\manset'|$ interviews. Recall that $(\man,\matching(\man)) \notin V$ if $\man \in \manset'$. Therefore none of the interviews we have accounted for so far, a total of $|PBP|+|\manset'|$ interviews, correspond to a vertex in $C$.
\end{proof}

Theorem~\ref{thm:reduction-to-VC} essentially tells us that an instance $(\inputsetting,\pref_{\manset,\womanset},\matching)$ of \minicrx\ is polynomial-time solvable if \minvc\ is polynomial-time solvable in $G(\inputsetting,\matching)$. Equipped with this knowledge, we provide three different restricted settings under which \icrx, and hence \minicrx, is solvable in polynomial time.

\begin{theorem}\label{thm:one-side-full}
\minicrx\ is solvable in polynomial time if one side has fully known strict preference ordering.
\end{theorem}
\begin{proof}
Assume that women have strict preferences and the target matching is $\matching$. Note that all PBPs must be of degree 1. Therefore $G(\inputsetting,\matching)$ is an empty graph with vertex cover of size zero. It follows from Proposition~\ref{prop:pbp} and Lemma~\ref{lem:type2} that \minicr\ has a solution of size $|PBP|+|\manset'|$.
\end{proof}

\begin{theorem}\label{thm:atmost-three}
\minicrx\ is solvable in polynomial time under the restriction of \smti\ in which indifference classes are of size at most 2.
\end{theorem}
\begin{proof}
We show that $G(\inputsetting,\matching)$  is a collection of cycles and paths, and hence its minimum vertex cover can be computed in polynomial time. The size of a minimum vertex cover for any path or cycle of length $\ell$ is $\myceil{\frac{\ell}{2}}$.

Take any vertex $v_1 = (\man,\matching(\man))$ in $V$. Recall that if any vertex $v_2  = (\man',\matching(\man'))$ is a neighbor of $v_1$, then it must be that at least one of $(\man,\matching(\man'))$ or $(\man',\matching(\man))$ is in $PBP'_2$. Note that if $(\man,\matching(\man')) \in PBP'_2$, then under $\inputsetting$ man $\man$ is indifferent between $\matching(\man)$ and $\matching(\man')$. Since each indifference class is of size at most 2, at most one such neighbor exists. Likewise, if $(\man',\matching(\man)) \in PBP'_2$ then $\matching(\man)$ is indifferent between $\man$ and $\man'$. However, since each indifference class is of size at most 2, at most one such neighbor exist. Thus, each vertex has degree at most 2, hence $G(\inputsetting,\matching)$  is a collection of cycles and paths.
\end{proof}

\begin{theorem}\label{thm:common-info}
\minicrx\ is solvable in polynomial time under the restriction of \smti\ in which all men are endowed with the same indifference classes, as well as all women. That is $\eqclass_i^{\man}=\eqclass_i^{\man'}$ for all $\man, \man' \in \manset$ and all $i\in[\numofwomen]$, and $\eqclass_i^{\woman}=\eqclass_i^{\woman'}$ for all $\woman, \woman' \in \womanset$ and all $i\in[\numofmen]$.
\end{theorem}
\begin{proof}
We show that $G(\inputsetting,\matching)$ is a collection of complete graphs, and hence its minimum vertex cover can be computed in polynomial time, since the size of a minimum vertex cover for any complete graph $K_{\ell}$ is equal to $\ell-1$. To prove that $G(\inputsetting,\matching)$ is a collection of complete graphs, we show that for any three given vertices $v_1, v_2$ and $v_3$, if $(v_1,v_2) \in E$ and $(v_1,v_3)\in E$ then $(v_2,v_3)\in E$. 


Take any three vertices $v_1 = (\man,\matching(\man))$, $v_2 = (\man',\matching(\man'))$, and $v_3 = (\man'',\matching(\man''))$. If $(v_1,v_2) \in E$ then, under $\inputsetting$, all men are indifferent between $\matching(\man)$ and $\matching(\man')$, all women are indifferent between $\man$ and $\man'$, and $\man, \man'\notin M'$. If $(v_1,v_3) \in E$ then, under $\inputsetting$, all men are indifferent between $\matching(\man)$ and $\matching(\man'')$, all women are indifferent between $\man$ and $\man''$, and $\man''\notin M'$. Therefore, since $\inputsetting$ is an instance of SMTI, all men are indifferent between $\matching(\man)$, $\matching(\man')$ and $\matching(\man'')$, and all women are indifferent between $\man$, $\man'$, and $\man''$. Hence $(\man',\matching(\man''))$ and $(\man'',\matching(\man'))$ are PBPs. If $(\man',\matching(\man''))$ is a PBP-D2 then, as $\man', \man'' \notin M'$, $(\man',\matching(\man''))\in PBP'_2$ and therefore $(v_2,v_3)\in E$. Assume for a contradiction that $(\man',\matching(\man''))$ is a PBP-D1. Assume that $\matching(\man'') \pref_{\man'} \matching(\man')$ (the argument is similar if $\man' \pref_{\matching(\man'')} \man''$), implying that $PBP_1(\matching(\man'')) \neq \emptyset$ and thus $\man''\in M'$, a contradiction.
\end{proof}

Theorem \ref{thm:atmost-three} is likely to be of more theoretical interest.  For the setting of Theorem \ref{thm:one-side-full}, we could envisage a hospitals-residents matching problem where residents are ranked uniformly (i.e., in a "master list" common to all hospitals \cite{scott-master}) according to some known objective value (e.g., which may be based on academic merit, as in the UK) and residents must use interviews in order to determine their true preferences over acceptable hospitals. For the setting of Theorem \ref{thm:common-info}, consider a market with ``tiered'' preferences, where everybody agrees who/what belongs to each tier (again the membership of these tiers could relate to some objective values), but the precise ordering within these tiers could be subjective, and up to individuals to determine themselves. For example, students may use national league tables for determining top tier universities, second tier universities and so on, but students' precise ranking over the universities in any given tier may vary.

If $\inputsetting$ is of one of the restricted forms for which \minicrx\ is polynomial time solvable, then one straightforward approach to solving \minicr\ is to enumerate all matchings that are stable under $\pref_{\manset,\womanset}$ and then solve \minicrx\ for each of them. This approach is practical if $\pref_{\manset,\womanset}$ admits a polynomial number of stable matchings.

%
\section{Conclusion and Future Work}
\label{sec:conc}
In this paper we have studied the complexity of the offline problem relating to computing an optimal interview strategy for a stable marriage market where initially participants have incomplete information, and the aim is to refine the instance using the minimum number of interviews in order to arrive at a super-stable matching.  The main direction for future work is to investigate the online case, where the true underlying preferences are not known to the mechanism designer, with respect to measures such as the competitive ratio.
Furthermore, an important question for which we do not know an answer yet is whether \minicr\ is polynomial-time solvable under some restricted setting. Extending the known results on interviewing in stable marriage markets to many-to-one markets such as college admission is another important future direction. It is also interesting to study online algorithms in a setting where elicitation is taking place via comparison queries. 
In this paper we assume that the objective of the mechanism designer is to minimize the total number of interviews overall. One may however argue that such a strategy may require one or may agents to conduct most of the interviews while the others do none or very little. In the view of fairness and the practicality of such central interview-scheduling schemes, it is also of utmost importance to study settings in which a fair distribution of the interviews is also considered.

\bibliographystyle{abbrv}
\bibliography{bibfile}

\end{document}